%% file: main.tex
\documentclass[a4paper]{article}
\usepackage[margin=1in]{geometry}
\usepackage{enumerate}
\usepackage{amssymb}
\usepackage{amsthm}
\usepackage{amsmath}
\usepackage{graphicx}
\usepackage{stmaryrd}

\usepackage{url}
\usepackage{rotating}
\usepackage{comment}
\usepackage{wrapfig}

\usepackage[algoruled,linesnumbered,algo2e,vlined]{algorithm2e}
\SetArgSty{textrm}

\newtheorem{theorem}{Theorem}
\newtheorem{definition}[theorem]{Definition}
\newtheorem{lemma}[theorem]{Lemma}
\newtheorem{corollary}[theorem]{Corollary}

\newtheorem{example}[theorem]{Example}
\newtheorem{problem}[theorem]{Problem}

\input{macro}

\title{
  Shortest unique palindromic substring queries in optimal time
}
\author{
  Yuto~Nakashima$^{1, 2}$\quad
  Hiroe~Inoue$^1$\quad
  Takuya~Mieno$^1$\quad\\
  Shunsuke~Inenaga$^1$\quad
  Hideo~Bannai$^1$\quad
  Masayuki~Takeda$^1$\\
  {$^1$ Department of Informatics, Kyushu University}\\
  {$^2$ Japan Society for the Promotion of Science (JSPS), Japan}\\
  {\texttt{\{yuto.nakashima,hiroe.inoue,takuya.mieno,}}\\
  {\texttt{inenaga,bannai,takeda\}@inf.kyushu-u.ac.jp}}
}

\date{}

\begin{document}

\maketitle

\begin{abstract}
A palindrome is a string that reads the same forward and backward.
A palindromic substring $P$ of a string $S$ 
is called a shortest unique palindromic substring ($\sups$) for an interval $[s, t]$ in $S$,
if $P$ occurs exactly once in $S$, this occurrence of $P$ contains interval $[s, t]$,
and every palindromic substring of $S$ which contains interval $[s, t]$ and is shorter than $P$ occurs at least twice in $S$.
The $\sups$ problem is, given a string $S$, to preprocess $S$ so that for any subsequent query interval $[s, t]$ 
all the $\supss$ for interval $[s, t]$ can be answered quickly.
We present an optimal solution to this problem.
Namely, we show how to preprocess a given string $S$ of length $n$ in $O(n)$ time and space so that all $\supss$
for any subsequent query interval can be answered in $O(\alpha+1)$ time, where $\alpha$ is the number of outputs.
\end{abstract}

\input{introduction}
\input{preliminaries}

\input{sups}
\input{mups}
\input{conclusion}

\bibliographystyle{abbrv}
\bibliography{ref}

\end{document}

%% file: macro.tex
\newcommand{\sus}{\mathit{SUS}}
\newcommand{\suss}{\mathit{SUS}\mbox{s}}

\newcommand{\rev}[1]{{#1}^{\mathit{R}}}
\newcommand{\occ}{\mathit{occ}}
\newcommand{\sups}{\mathit{SUPS}}
\newcommand{\supss}{\mathit{SUPS}\mbox{s}}
\newcommand{\mups}{\mathit{MUPS}}
\newcommand{\mupss}{\mathit{MUPS}\mbox{s}}
\newcommand{\M}{\mathcal{M}}
\newcommand{\m}{\mathit{mups}}
\newcommand{\mupslen}{\mathit{Mlen}}

\newcommand{\predmups}{\mathit{predMUPS}}
\newcommand{\succmups}{\mathit{succMUPS}}
\newcommand{\predend}{\mathit{predMUPS}}
\newcommand{\succbeg}{\mathit{succMUPS}}
\newcommand{\distinctpal}{\mathit{DP}}
\newcommand{\rmq}{\mathit{RmQ}}

\newcommand{\ISA}{\mathit{ISA}}
\newcommand{\LCPA}{\mathit{LCP}}
\newcommand{\LPF}{\mathit{LPF}}



%% file: introduction.tex
\section{Introduction}
A substring $S[i..j]$ of a string $S$ is called a \emph{shortest unique substring} ($\sus$) for a position $p$
if $S[i..j]$ is the shortest substring
s.t. $S[i..j]$ is unique in $S$ (i.e., $S[i..j]$ occurs exactly once in $S$),
and $[i..j]$ contains $p$ (i.e., $i \leq p \leq j$).
Recently, Pei et al.~\cite{Pei} proposed the \emph{point SUS problem},
preprocessing a given string $S$ of length $n$ so that we can return a $\sus$ for any given query position efficiently.
This problem was considered for some applications in bioinformatics,
e.g., polymerase chain reaction (PCR) primer design in molecular biology.
Pei et al.~\cite{Pei} proposed an algorithm which returns a $\sus$ for any given position
in constant time after $O(n^2)$-time preprocessing.
After that, Tsuruta et al.~\cite{Tsuruta} and Ileri et al.~\cite{ileri14:_short_unique_subst_query_revis}
independently showed optimal $O(n)$-time preprocessing and constant query time algorithms.
They also showed optimal $O(n)$-time preprocessing and $O(k)$ query time algorithms
which return all $\sus\mbox{s}$ for any given position
where $k$ is the number of outputs.
Moreover, Hon et al.~\cite{HonTX15} proposed an in-place algorithm which returns a $\sus$.
A more general problem called \emph{interval SUS problem},
where a query is an interval, was considered by Hu et al.~\cite{Hu}.
They proposed an optimal $O(n)$-time preprocessing and $O(k)$ query time algorithm
which returns all $\suss$ containing a given query interval.
Most recently, Mieno et al.~\cite{MienoIBT16} proposed an efficient algorithm
for interval $\sus$ problem when the input string is represented by \emph{run-length encoding}.

In this paper, we consider a new variant of interval $\sus$ problems concerning palindromes.
A substring $S[i..j]$ is called a palindromic substring of $S$
if $S[i..j]$ and the reversed string of $S[i..j]$ is the same string.
The study of combinatorial properties and structures on palindromes is still an important and
well studied topic in stringology~\cite{DBLP:conf/dlt/BannaiGIKKPPS15,DBLP:journals/tcs/DroubayJP01,DBLP:journals/jda/FiciGKK14,DBLP:journals/ipl/GroultPR10,DBLP:conf/cpm/ISIBT14,eertree}.
Droubay et al.~\cite{DBLP:journals/tcs/DroubayJP01} showed a string of length $n$ can contain 
at most $n+1$ distinct palindromes.
Moreover, Groult et al.~\cite{DBLP:journals/ipl/GroultPR10} proposed a linear time algorithm
for computing all distinct palindromes in a string.

Our new problem can be described as follows.
A substring $S[i..j]$ of a string $S$ is called a \emph{shortest unique palindromic substring} ($\sups$) for an interval $[s, t]$
if $S[i..j]$ is the shortest substring
s.t. $S[i..j]$ is unique in $S$, $[i..j]$ contains $[s, t]$, and $S[i..j]$ is a palindromic substring.
The \emph{interval SUPS problem} is to preprocess a given string $S$ of length $n$
so that we can return all $\supss$ for any query interval efficiently.
For this problem, we propose an optimal $O(n)$-time preprocessing and $O(\alpha+1)$-time query algorithm,
where $\alpha$ is the number of outputs.
Potential applications of our algorithm are in bioinformatics;
It is known that the presence of particular (e.g., unique) palindromic sequences can affect
immunostimulatory activities of oligonucleotides~\cite{kuramoto1992oligonucleotide,yamamoto1992unique}.
The size and the number of palindromes also influence the activity.
Since any unique palindromic sequence can be obtained easily from a shorter unique palindromic sequences,
we can focus on the shortest unique palindromic substrings.

The contents of our paper are as follows.
In Section~\ref{sec:preliminaries}, we state some definitions and properties on strings.
In Section~\ref{sec:sups}, we explain properties on $\sups$ and our query algorithm.
In Section~\ref{sec:mups}, we show the main part of the preprocessing phase of our algorithm.
Finally, we conclude.

%% file: preliminaries.tex
\section{Preliminaries} \label{sec:preliminaries}

\subsection{Strings}
Let $\Sigma$ be an integer {\em alphabet}.
An element of $\Sigma^*$ is called a {\em string}.
The length of a string $S$ is denoted by $|S|$. 
The empty string $\varepsilon$ is a string of length 0,
namely, $|\varepsilon| = 0$.
Let $\Sigma^+$ be the set of non-empty strings,
i.e., $\Sigma^+ = \Sigma^* - \{\varepsilon\}$.
For a string $S = xyz$, $x$, $y$ and $z$ are called
a \emph{prefix}, \emph{substring}, and \emph{suffix} of $S$, respectively.
A prefix $x$ and a suffix $z$ of $S$ are respectively 
called a \emph{proper prefix} and \emph{proper suffix} of $S$, 
if $x \neq S$ and $z \neq S$.
The $i$-th character of a string $S$ is denoted by $S[i]$, where $1 \leq i \leq |S|$.
For a string $S$ and two integers $1 \leq i \leq j \leq |S|$, 
let $S[i..j]$ denote the substring of $S$ that begins at position $i$ and ends at position $j$.
For convenience, let $S[i..j] = \varepsilon$ when $i > j$.

\subsection{Palindromes}
Let $\rev{S}$ denote the reversed string of $S$,
that is, $\rev{S} = S[|S|] \cdots S[1]$.
A string $S$ is called a palindrome if $S = \rev{S}$.
Let $P \subset \Sigma^*$ be the set of palindromes.
A substring $S[i..j]$ of $S$ is said to be a palindromic substring of $S$, if $S[i..j] \in P$.
The center of a palindromic substring $S[i..j]$ of $S$ is $\frac{i+j}{2}$.
Thus a string $S$ of length $n \geq 1$ has $2n-1$ centers
($1, 1.5, \ldots, n-0.5, n$).
The following lemma can be easily obtained by the definition of palindromes.

\begin{lemma} \label{lem:rev_in_pal}
Let $S$ be a palindrome.
For any integers $i, j$ s.t. $1 \leq i \leq j \leq |S|$,
$S[|S|-j+1..|S|-i+1] = \rev{S[i..j]}$ holds.
\end{lemma}

\subsection{$\mupss, \supss$ and our problem}
For any non-empty strings $S$ and $w$,
let $\occ_S(w)$ denote the set of occurrences of $w$ in $S$,
namely, $\occ_S(w) = \{i \mid 1 \leq i \leq |S|-|w|+1, w = S[i..i+|w|-1]\}$.
A substring $w$ of a string $S$ is called a \emph{unique substring} 
(resp. a \emph{repeat}) of $S$
if $|\occ_S(w)| = 1$ (resp. $|\occ_S(w)| \geq 2$).
In the sequel, we will identify each unique substring $w$ of $S$
with its corresponding (unique) interval 
$[i, j]$ in $S$ such that $w = S[i..j]$.
A substring $S[i..j]$ is said to be \emph{unique palindromic substring}
if $S[i..j]$ is a unique substring in $S$ and a palindromic substring.
We will say that an interval $[i_1, j_1]$ contains an interval $[i_2, j_2]$
if $i_1 \leq i_2 \leq j_2 \leq j_1$ holds.
The following notation is useful in our algorithm.

\begin{definition}[Minimal Unique Palindromic Substring($\mups$)]
A string $S[i..j]$ is a $\mups$ in $S$
if $S[i..j]$ satisfies all the following conditions;
\begin{itemize}
\item $S[i..j]$ is a unique palindromic substring in $S$,
\item $S[i+1..j-1]$ is a repeat in $S$ or $1 \leq |S[i..j]| \leq 2$.
\end{itemize}
\end{definition}

Let $\M_S$ denote the set of intervals of all $\mupss$ in $S$
and let $\m_i = [b_i, e_i]$ denote the $i$-th $\mups$ in $\M_S$
where $1 \leq i \leq m$ and $m$ is the number of $\mupss$ in $S$.
We assume that $\mupss$ in $\M_S$ are sorted in increasing order of beginning positions.
For convenience, we define $\m_0 = [-1, -1], \m_{m+1} = [n+1, n+1]$.
\begin{example}[$\mups$]
For $S = \mathtt{acbaaabcbcbcbaab}$, $\M_S = \{ [4, 6], [8, 12], [13, 16] \}$
(see also Fig.~\ref{fig:example}).
\end{example}

\begin{definition}[Shortest Unique Palindromic Substring($\sups$)]
A string $S[i..j]$ is a $\sups$ for an interval $[s, t]$ in $S$
if $S[i..j]$ satisfies all the following conditions;
\begin{itemize}
\item $S[i..j]$ is a unique palindromic substring in $S$,
\item $[i, j]$ contains $[s, t]$,
\item no unique palindromic substring $S[i'..j']$ containing $[s, t]$ 
	  with $j'-i' < j-i$ exists.
\end{itemize}
\end{definition}

\begin{example}[$\sups$]\label{exp:sups}
Let $S = \mathtt{acbaaabcbcbcbaab}$.
$\sups$ for interval $[6, 7]$ is the $S[3..7] = \mathtt{baaab}$.
$\sups$ for interval $[7, 8]$ are $S[2..8] = \mathtt{cbaaabc}$ and $S[7..13] = \mathtt{bcbcbcb}$.
$\sups$ for interval $[4, 13]$ does not exist.
(see also Fig.~\ref{fig:example}).
\end{example}

\begin{figure}[tbh]
  \centerline{
    \includegraphics[width = 1.0\textwidth]{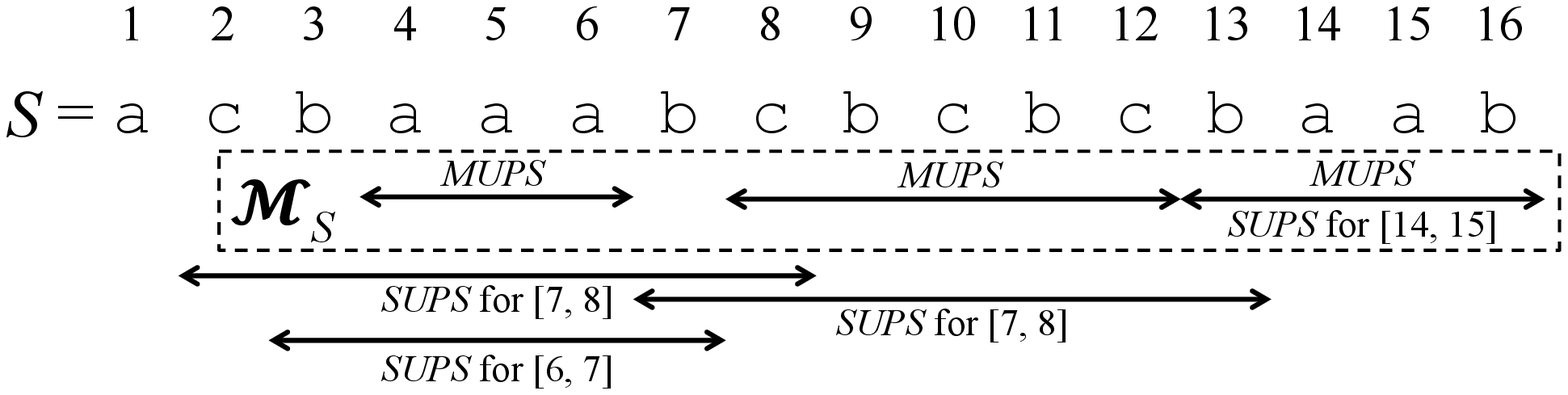}
  }
  \caption{
  	This figure shows all $\mupss$ for $S = \mathtt{acbaaabcbcbcbaab}$ 
  	and some $\sups$ described in Example~\ref{exp:sups}.
  }
  \label{fig:example}
\end{figure}

In this paper, we tackle the following problem.

\begin{problem}[$\sups$ problem]
\leavevmode \par
\begin{itemize}
\item {\bf Preprocess} : String $S$ of length $n$.
\item {\bf Query} : An interval $[s, t] (1 \leq s \leq t \leq n)$.
\item {\bf Return} : All the $\supss$ for interval $[s, t]$.
\end{itemize}
\end{problem}

\subsection{Computation Model}
Our model of computation is the word RAM:
We shall assume that the computer word size is at least $\lceil \log_2 n \rceil$, 
and hence, standard operations on
values representing lengths and positions of strings 
can be manipulated in constant time.
Space complexities will be determined by the number of computer words (not bits).

%% file: sups.tex
\section{Solution to the $\sups$ problem} \label{sec:sups}
In this section, we show how to compute all $\supss$ for any query interval $[s, t]$.

\subsection{Properties on $\sups$ and $\mups$} \label{subsec:property}
In our algorithm, we compute $\supss$ by using $\mupss$.
Firstly, we show the following lemma.
Lemma~\ref{lem:non_nesting} states that $\mupss$ cannot nest in each other.

\begin{lemma} \label{lem:non_nesting}
For any pair of distinct $\mupss$, one cannot contain the other.
\end{lemma}

\begin{proof}
  Consider two $\mupss$ $u,v$ such that $u$ contains $v$.
  If $u$ and $v$ have the same center, then $u$ is not a $\mups$.
  On the other hand, if $u$ and $v$ have a different center, 
  we have from Lemma~\ref{lem:rev_in_pal} and that $v$ is a palindromic substring,
  $v$ occurs in $u$ at least twice.
  This contradicts that $v$ is unique.
\end{proof}

From this lemma, we can see that no pair of distinct $\mupss$ begin nor end at the same position.
This fact implies that the number of $\mupss$ is at most $n$ for any string of length $n$.
The following lemma states a characterization of $\supss$ by $\mupss$.

\begin{lemma} \label{lem:unimups_in_sups}
  For any $\sups$ $S[i..j]$ for some interval,
  there exists exactly one $\mups$ that is contained in $[i,j]$.
  Furthermore, the $\mups$ has the same center as $S[i..j]$.
\end{lemma}

\begin{proof}
Let $S[i..j]$ be a $\sups$ for some interval.
$S[i..j]$ contains a $\mups$ $S[x_1..y_1]$ of the same center,
i.e., $\frac{i+j}{2} = \frac{x_1+y_1}{2}$, s.t. $j-i \geq y_1-x_1$.
Suppose that there exists another $\mups$ $S[x_2..y_2]$ contained in $[i, j]$.
From Lemma~\ref{lem:non_nesting}, $S[x_1..y_1]$ and $S[x_2..y_2]$ do not have the same center.
On the other hand, if $S[x_1..y_1]$ and $S[x_2..y_2]$ have different centers,
then $S[x_2..y_2]$ occurs at least two times in $S[i..j]$ by Lemma~\ref{lem:rev_in_pal},
since $S[x_2..y_2] = \rev{S[x_2..y_2]}$.
This contradicts that $S[x_2..y_2]$ is a $\mups$.
\end{proof}

From the above lemma,
any $\sups$ contains exactly one $\mups$ which has the same center 
(see also Fig.~\ref{fig:example}).
Below, we will describe the relationship
between a query interval $[s, t]$ and the $\mups$ contained in a $\sups$ for $[s, t]$.
Before explaining this, we define the following notations.

\begin{itemize}
\item $\M([s, t])$ : the set of $\mupss$ containing $[s, t]$.
\item $\predmups[t] = i$ s.t. $i = \max \{ k \mid e_k \leq t \}$.
\item $\succmups[s] = i$ s.t. $i = \min \{ k \mid s \leq b_k \}$.
\end{itemize}
In other words, $\m_{\predmups[t]}$ is the rightmost $\mups$
which ends before position $t+1$,
and $\m_{\succmups[s]}$ is the leftmost $\mups$
which begins after position $s-1$.

\begin{lemma} \label{lem:candidates}
Let $S[i..j]$ be a $\sups$ for an interval $[s, t]$.
Then, the unique $\mups$ $S[x..y]$ 
contained in $[i,j]$ is in
$\{ \predmups[t] \} \cup \M([s, t]) \cup \{\succmups[s] \}$.
\end{lemma}

\begin{proof}
  Assume to the contrary that there exists a $\sups$ $S[i..j]$ that contains a
  $\mups$ $S[x..y] \notin \{\predmups[t]\} \cup \M([s, t]) \cup \{\succmups[s]\}$.
Since $S[x..y]\not\in\M([s,t])$, $[x,y]$ does not contain $[s,t]$. Thus, there can be the following two cases:
\begin{itemize}
\item If $y < t$, there must exist $\mups$ $[x', y']$ s.t. $y < y' \leq t$, since $S[x,y] \neq \predmups[t]$.
By Lemma~\ref{lem:non_nesting}, $x < x'$.
Thus $i \leq x < x' \leq y' \leq t \leq j$ holds.
However, this contradicts Lemma~\ref{lem:unimups_in_sups}.
\item If $s < x$, there must exist $\mups$ $[x', y']$ s.t. $s \leq x' < x$, since $S[x,y] \neq \succmups[s]$.
By Lemma~\ref{lem:non_nesting}, $y' \leq y$.
Thus $i \leq s \leq x' \leq y' \leq y \leq j$ holds,
However, this contradicts Lemma~\ref{lem:unimups_in_sups}.
\end{itemize}
Therefore the lemma holds.

\end{proof}

Next, we want to explain how $\supss$ are related to $\mupss$.
It is easy to see that there may not be a $\sups$ for some query interval.
We first show a case where there are no $\supss$ for a given query.
The following corollary is obtained from Lemma~\ref{lem:unimups_in_sups}.

\begin{corollary} \label{coro:no_answer}
Let $S[x_1..y_1]$ and $S[x_2..y_2]$ be $\mupss$ contained in a query interval $[s, t]$.
There is no $\sups$ for an interval $[s, t]$.
\end{corollary}

From this corollary, a $\sups$ for an interval $[s, t]$ can exist
if the number of $\mupss$ contained in $[s, t]$ is less than or equal to 1.
The following two lemmas show what the $\sups$ for $[s, t]$ is, when $[s, t]$ contains only one $\mups$,
and when $[s, t]$ does not contain any $\mupss$.

\begin{lemma} \label{contain_one_mups}
Let $S[x..y]$ be the only $\mups$ contained in the query interval $[s, t]$.
If $S[x-z, y+z]$ is a palindromic substring where $z = \max \{ x-s, t-y \}$,
then $S[x-z, y+z]$ is the $\sups$ for $[s, t]$.
Otherwise, there is no $\sups$ for $[s, t]$.
\end{lemma}

\begin{proof}
Assume that there exists a $\sups$ $u$ for $[s, t]$
which has the same center with a $\mups$ other than $S[x..y]$.
By the definition of $\sups$, $u$ should contain $[s, t]$.
Since $[s, t]$ contains $[x, y]$,
$u$ contains two $\mupss$, a contradiction.
Thus, there can be no $\sups$ s.t. the center is not $\frac{x+y}{2}$.
It is clear that $S[x-z, y+z]$ is a unique palindromic substring
if $S[x-z, y+z]$ is a palindromic substring where $z = \max \{ x-s, t-y \}$.
Therefore the lemma holds.

\end{proof}

\begin{lemma} \label{lem:answers}
Let $[s, t]$ be the query interval.
Then $\supss$ for $[s, t]$ are the shortest of the following candidates.
\begin{enumerate}
\item $S[x..y]$ s.t. $[x, y] \in \M([s, t])$,
\item $S[x-t+y..t]$ s.t. $[x, y] = \predmups([s, t])$, if it is a palindromic substring,
\item $S[s..y+x-s]$ s.t. $[x, y] = \succmups([s, t])$, if it is a palindromic substring.
\end{enumerate}
\end{lemma}

\begin{proof}
It is clear that $S[x..y]$ is a unique palindromic substring containing $[s, t]$
if $[x, y] \in \M([s, t])$ exists.
It is also clear that
if $[x, y] = \predmups([s, t])$ or $[x, y] = \succmups([s, t])$,
then $S[x-t+y..t]$ or $[s..y+x-s]$, respectively, are unique palindromic substrings,
if they are palindromic substrings.
By Lemma~\ref{lem:candidates}, we do not need to consider palindromic substrings which
have the same center as $\mupss$ other than the candidates considered above.
Thus the shortest of the candidates is $\sups$ for $[s, t]$
(see also Fig.~\ref{fig:sups_candidate}).
\end{proof}

\begin{figure}[tbh]
  \centerline{
    \includegraphics[width = 1.0\textwidth]{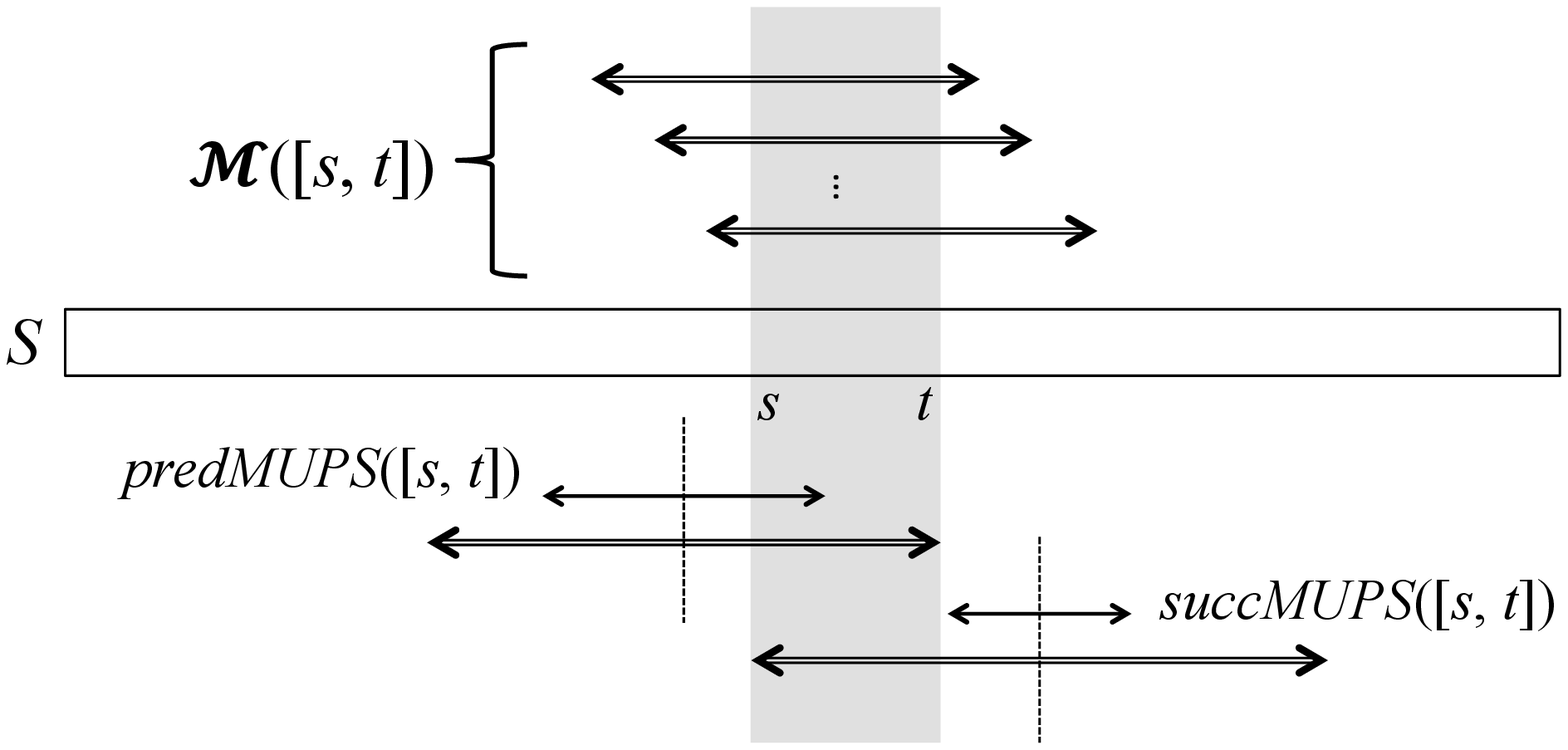}
  }
  \caption{
    Double arrows represent the candidates of $\sups$ for $[s, t]$.
    The shortest of the candidates is $\sups$ for $[s, t]$.
  }
  \label{fig:sups_candidate}
\end{figure}

From the above arguments, 
the number of $\mupss$ is useful to compute $\supss$ for a query interval.
The following lemma shows how to compute the number of $\mupss$ contained in a given interval.

\begin{lemma} \label{lem:catching_mups}
For any interval $[s, t]$,
\begin{itemize}
\item if $\succbeg[s] > \predend[t]$, $[s, t]$ contains no $\mups$,
\item if $\succbeg[s] = \predend[t]$, $[s, t]$ contains only one $\mups$, \\$\m_{\succbeg[s]} = \m_{\predend[t]}$, and
\item if $\succbeg[s] < \predend[t]$, $[s, t]$ contains at least two $\mupss$.
\end{itemize}
\end{lemma}

\begin{proof}
\leavevmode \par
\begin{itemize}
\item Let $j = \succbeg[s] > \predend[t] = i$.
Then $b_i < s \leq b_j$ and $e_i \leq t < e_j$ hold,
and thus neither of $\m_i$ and $\m_j$ are contained in $[s, t]$.
If we assume that $[s, t]$ contains a $\mups$ $\m_k$ for some $k$,
it should be that $i < k < j$, $b_i < s \leq b_k < b_j$.
However, this contradicts that $j = \succbeg[s]$
(see also the top in Fig.~\ref{fig:mups_num.eps}).
\item Let $\succbeg[s] = \predend[t] = i$.
Since $\succbeg[s] = i$, $b_{i-1}$ should be less than $s$, and $b_i$ at least $s$.
Since $\predend[t] = i$, $e_{i+1}$ should be larger than $t$, and $e_i$ at most $t$.
Thus $[s, t]$ only contains $\m_i$
(see also the middle in Fig.~\ref{fig:mups_num.eps}).
\item Let $i = \succbeg[s] < \predend[t] = j$.
Then $s \leq b_i < b_j$ and $e_i < e_j \leq t$ hold,
which implies $s \leq b_i \leq e_i < t$ and $s < b_j \leq e_j \leq t$.
Thus, both $\m_i$ and $\m_j$ are contained in $[s, t]$
(see also the bottom in Fig.~\ref{fig:mups_num.eps}).
\end{itemize}
\end{proof}
\begin{figure}[h!]
  \centerline{
    \includegraphics[width = 1.0\textwidth]{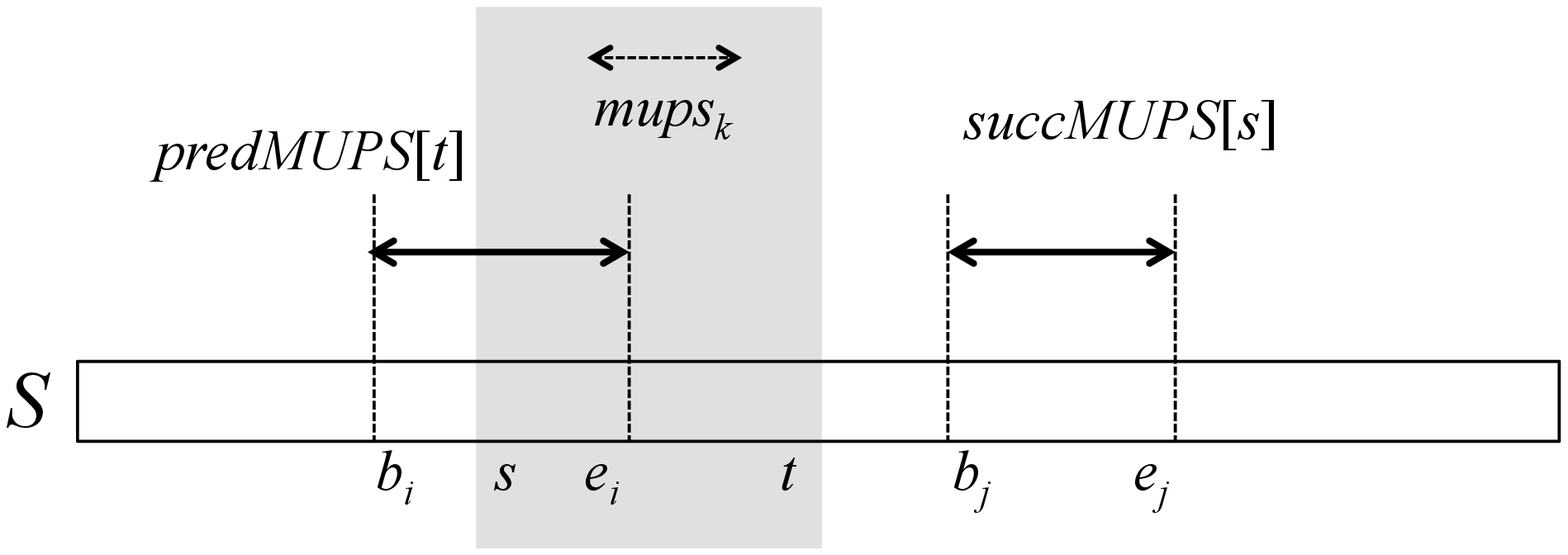}
  }
  \centerline{
    \includegraphics[width = 1.0\textwidth]{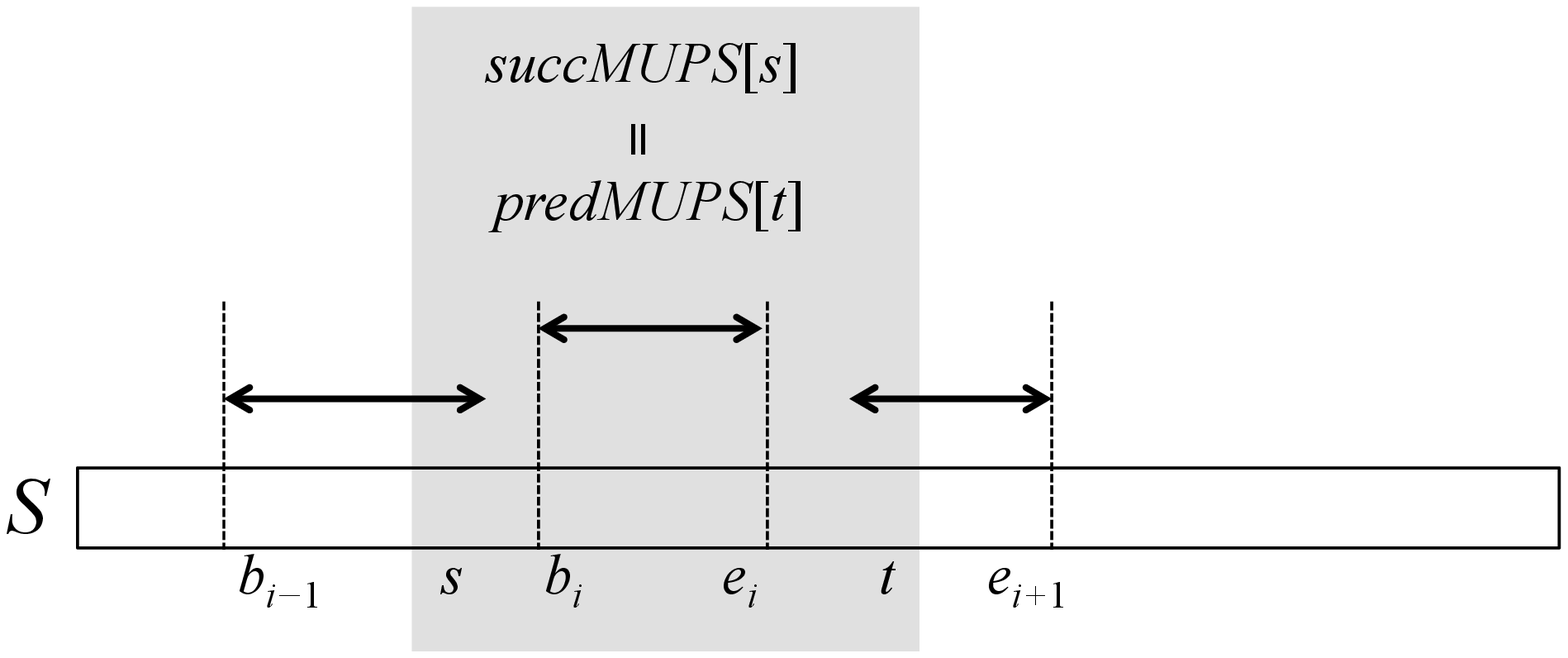}
  }
  \centerline{
    \includegraphics[width = 1.0\textwidth]{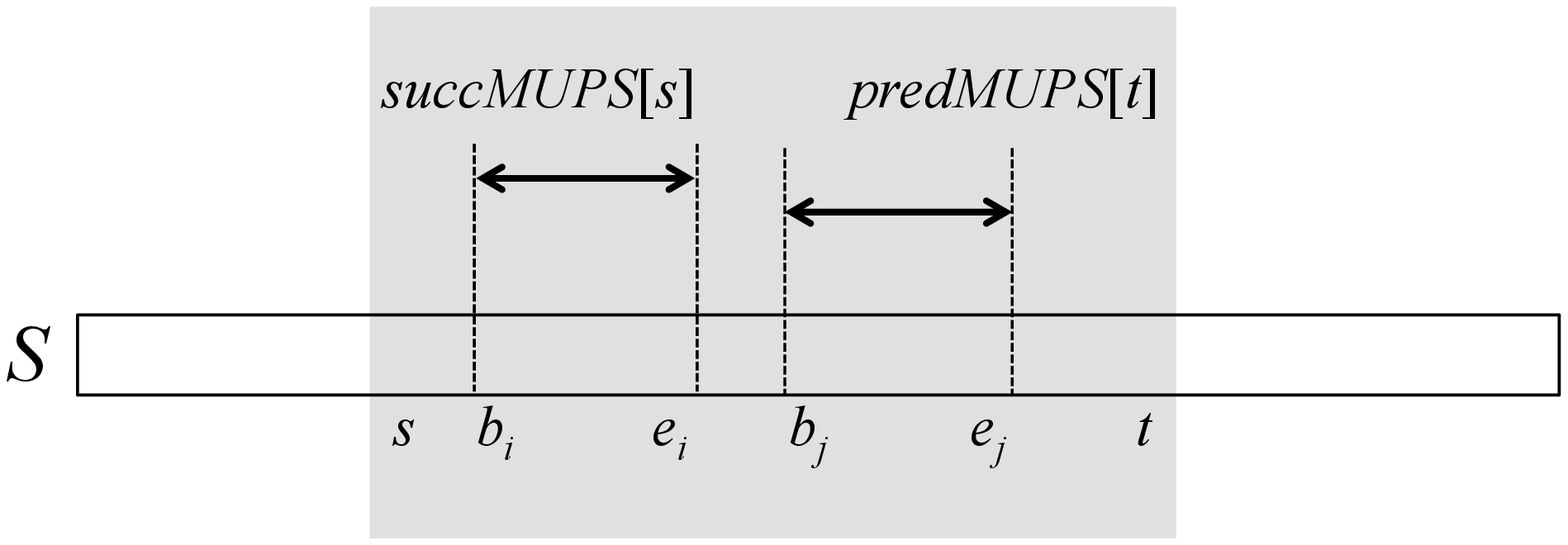}
  }
  \caption{
    Illustrations for proof of Lemma~\ref{lem:catching_mups}.
  }
  \label{fig:mups_num.eps}
\end{figure}

\subsection{Tools}
\label{subsec:tools}
Here, we show some tools for our algorithm.

\begin{lemma}[e.g.,~\cite{eertree}] \label{lem:palindromic_checking}
  For any interval $[i, j]$ in $S$ of length $n$, we can check whether $S[i..j]$
  is a palindromic substring or not in $O(n)$ preprocessing time and constant query time
  with $O(n)$ space.
\end{lemma}
Manacher's algorithm~\cite{Manacher75} can compute all maximal palindromic substrings in linear time.
If we have the array of radiuses of maximal palindromic substrings for all $2n-1$ centers, 
we can check whether a given substring $S[i..j]$ is a palindromic substring or not in constant time.

\subsubsection{Range minimum queries (RmQ)}
Let $A$ be an integer array of size $n$.
A \emph{range minimum query} $\rmq_A(i, j)$ returns 
the index of a minimum element in the subarray $A[i, j]$
for given a query interval $[i, j] (1 \leq i \leq j \leq n)$,
i.e., it returns one of $\arg \min_{i \leq k \leq j}\{A[k]\}$.
It is well-known (see e.g.,~\cite{rmqspace}) that 
after an $O(n)$-time preprocessing over the input array $A$,
$\rmq_A(i, j)$ can be answered in $O(1)$ time 
for any query interval $[i, j]$, using $O(n)$ space.

\subsection{Algorithm}
Due to the arguments in Section~\ref{subsec:property},
if we can compute $\predmups$, the shortest $\mupss$ in $\M([s, t])$ and $\succmups$ for a query interval $[s, t]$,
then, we can compute $\supss$ for $[s, t]$.
Below, we will describe our solution to the $\sups$ problem.

\subsubsection{Preprocessing phase}
First, we compute $\M_S$ for a given string $S$ of length $n$ in increasing order of beginning positions.
We show, in the next section, that this can be done in $O(n)$ time and space.
After computing $\M_S$, we compute the arrays $\predend$ and $\succbeg$.
It is easy to see that we can also compute these arrays in $O(n)$ time by using $\M_S$.
In the query phase, we are required to compute the shortest $\mupss$ that  contain the query interval $[s, t]$.
To do so efficiently, we prepare the following array.
Let $\mupslen$ be an array of length $m = |\M_S|$,
and the $i$-th entry $\mupslen[i]$ holds the length of $\m_i$, i.e., $\mupslen[i] = |\m_i| = e_i-b_i+1$.
We also preprocess $\mupslen$ for $\rmq$ queries.
This can be done in $O(m)$ time and space as noted in Section~\ref{subsec:tools}.
Thus, since $m = O(n)$, the total  preprocessing is $O(n)$ time and space.

\subsubsection{Query phase}
First, we compute how many $\mupss$ are contained in a query interval $[s, t]$
by using Lemma~\ref{lem:catching_mups}, which we denote by $\mathit{num}$.
This can be done in $O(1)$ time given arrays $\predend$ and $\succbeg$.

\begin{itemize}
\item
  If $\mathit{num} = 0$,
  let $\m_i = \predmups([s, t])$ and $\m_j = \succmups([s, t])$, i.e.,
  $i = \predend[s]$ and $j = \succbeg[t]$.
  We check whether $S[b_i-t+e_i..t]$ and $S[s..e_j+b_j-s]$ are palindromic substrings or not.
  If so, then they are candidates of $\supss$ for $[s, t]$ by Lemma~\ref{lem:answers}.
  Let $q$ be the length of the shortest candidates which can be found in the above.
  Second, we compute the shortest $\mups$ in $\M([s, t])$,
  if their lengths are at most $q$.
  In other words, we compute the smallest values in $\mupslen[i+1..j-1]$, if they are at most $q$.
  We can compute all such $\mupss$ in linear time w.r.t. the number of such $\mupss$
  by using $\rmq$ queries on $\mupslen[i+1, j-1]$;
  if $k = \rmq_\mupslen(i+1,j-1)$ and $\mupslen[k] \leq q$, then we consider the range
  $\mupslen[i+1..k-1]$ and $\mupslen[k+1,j-1]$ and recurse. Otherwise, we stop the recursion.
  Finally, we return the shortest candidates as $\sups$.

\item If $\mathit{num} = 1$, let $\m_i$ be the $\mups$ contained in $[s, t]$.
  First, we check whether $S[b_i-z, e_i+z]$ is a palindromic substring or not
  by using Lemma~\ref{lem:palindromic_checking} where $z = \max \{ b_i-s, t-e_i \}$.
  If so, then return $[b_i-z, e_i+z]$,
  otherwise $\sups$ for $[s, t]$ does not exist.

\item If $\mathit{num} \geq 2$, 
  then, from Corollary~\ref{coro:no_answer}, $\sups$ for $[s, t]$ does not exist.  
\end{itemize}
Therefore, we obtain the following.

\begin{theorem}
After constructing an $O(n)$-space data structure of a given string of length $n$ in $O(n)$ time,
we can compute all $\supss$ for a given query interval $[s, t]$
in $O(\alpha+1)$ time where $\alpha$ is the number of outputs.
\end{theorem}

%% file: mups.tex
\section{Computing $\mupss$} \label{sec:mups}
In this section, we show how to compute $\M_S$ in $O(n)$ time and space.
Let $\distinctpal_S$ be the set of \emph{distinct palindromic substrings} in $S$,
and $\mathit{strM}_S = \{ S[i, j] \mid [i, j] \in \M_S \}$.
Our idea of computing $\M_S$ is based on the following lemma.

\begin{lemma}
  $\mathit{strM}_S \subseteq \distinctpal_S$.
\end{lemma}
\begin{proof}
  It is clear that any string in $\mathit{strM}_S$ is a palindromic substring of $S$.
\end{proof}

An algorithm for computing all distinct palindromic substrings in string in linear time and space
was proposed by Groult et al.~\cite{DBLP:journals/ipl/GroultPR10}.
We show a linear time and space algorithm which computes $\M_S$
by modifying Groult et al.'s algorithm.

\subsection{Tools}
We show some tools for computing $\M_S$ below.

\begin{itemize}
\item {\bf Longest previous factor array (LPF)}
We denote the longest previous factor array of $S$ by $\LPF_S$.
The $i$-th entry $(1 \leq i \leq n)$ is the length of the longest prefix of $S[i..n]$
which occurs at a position less than $i$.

\item {\bf Inverse suffix array (ISA)}
We denote the inverse suffix array of $S$ by $\ISA_S$.
The $i$-th entry $(1 \leq i \leq n)$ is the lexicographic order of $S[i..n]$
in all suffixes of $S$.

\item {\bf Longest common prefix array (LCP)}
We denote the longest common prefix array of $S$ by $\LCPA_S$.
The $i$-th entry $(2 \leq i \leq n)$ is the length of the longest common prefix
of the lexicographically $i$-th suffix of $S$ and the $(i-1)$-th suffix of $S$.
\end{itemize}

\subsection{Computing distinct palindromes}
Here, we show a summary of Groult et al.'s algorithm.
The following lemma states the main idea.

\begin{lemma}[\cite{DBLP:journals/tcs/DroubayJP01}]
The number of distinct palindromic substrings in $S$
is equal to the number of prefixes of $S$
s.t. its longest palindromic suffix is unique in the prefix.
\end{lemma}

Since counting suffixes that uniquely occur in a prefix implies that only
the leftmost occurrences of substrings, and thus distinct substrings are counted,
their algorithm finds all the distinct palindromic substrings by:
\begin{itemize}
\item computing the longest palindromic suffix of each prefix of $S$, and
\item checking whether each longest palindromic suffix occurs uniquely in the prefix or not.
\end{itemize}
They first propose an algorithm which computes all the longest palindromic suffixes in linear time.
They then check, in constant time, the uniqueness of the occurrence in the prefix by using the LPF array, thus
computing $\distinctpal_S$ in linear time and space.

\subsection{Computing all $\mupss$}
Finally, we show how to modify Groult et al.'s algorithm.
As mentioned, they compute the leftmost occurrence of each distinct palindromic substring.
We call such a palindromic substring, the leftmost palindromic substring.
It is clear that if a leftmost palindromic substring $w$ is unique in $S$
and is a minimal palindromic substring, then $w$ is a $\mups$.
Thus, we add operations to check the uniqueness and minimality of each leftmost palindromic substring.
We can do these operations by using $\ISA$ and $\LCPA$ array.

Let $S[i..j]$ be a leftmost palindromic substring in $S$.
First, we check whether $S[i..j]$ is unique or not in $S$.
If $\ISA[i] = k$, $S[i..n]$ is the lexicographically $k$-th suffix of $S$.
$S[i..j]$ is unique in $S$ iff $\LCPA[k] < j-i+1$ and $\LCPA[k+1] < j-i+1$.
Thus we can check whether $S[i..j]$ is unique or not in constant time.
Finally, we check whether $S[i..j]$ is a minimal palindromic substring or not.
By definition, $S[i..j]$ is minimal palindromic substring if $j-i+1 \leq 2$,
i.e., $S[i..j]$ has no shorter unique palindromic substring.
If $j-i+1 > 2$, then we check whether $S[i+1..j-1]$ is unique or not
by using $\ISA$ and $\LCPA$ in a similar way.
Thus we can also check whether $S[i..j]$ is minimal or not in constant time.
By the above arguments, we can compute all $\mupss$ in linear time and space.

%% file: conclusion.tex
\section{Conclusions} \label{sec:concl}
We consider a new problem called the shortest unique palindromic substring problem.
We proposed an optimal linear time preprocessing algorithm so that all $\supss$ for any given query interval
can be answered in linear time w.r.t. the number of outputs.
The key idea was to use palindromic properties in order to obtain a characterization of $\sups$,
more precisely,
that a palindromic substring cannot contain a unique palindromic substring with a different center.